\documentclass[11pt]{article}
\usepackage[margin=1in]{geometry}
\usepackage{amsmath,amssymb,amsthm,mathtools}
\usepackage{physics}
\usepackage{bm}
\usepackage{microtype}
\usepackage[dvipsnames]{xcolor}
\usepackage[hypertexnames=false,colorlinks=true,linkcolor=MidnightBlue,citecolor=ForestGreen,urlcolor=BrickRed]{hyperref}
\usepackage[nameinlink]{cleveref}
\numberwithin{equation}{section}
\newtheorem{theorem}{Theorem}[section]
\newtheorem{proposition}[theorem]{Proposition}
\newtheorem{corollary}[theorem]{Corollary}
\theoremstyle{remark}

\title{Global vs.\ Product Observables in Bipartite Quantum Systems: The Sharp Bound}
\author{Zhi Li\\{\normalsize IBM Research}\and Xiaofei Shi\\{\normalsize University of Toronto}}
\date{}
\begin{document}
\maketitle
\begin{abstract}
To probe a bipartite quantum system, one may use arbitrary global operators or restrict to product operators acting separately on the two subsystems. We determine the sharp universal comparison between the resulting norms. For every $z\in M_n\otimes M_m$, we prove
\begin{equation*}
  \norm{z}_1
  \le \sqrt{2}\,\min\{n,m\}\,\norm{z}_\varepsilon,
\end{equation*}
where $\norm{\cdot}_1$ is the trace norm and $\norm{\cdot}_\varepsilon$ is the injective tensor norm associated with the trace norms on $M_n$ and $M_m$. To prove the upper bound, we establish an $L_1$ noncommutative Khintchine inequality whose random coefficients are the entries of a Haar unitary. 
We also show that the coefficient $\sqrt{2}$ is sharp. 
As applications, we show that the same sharp constant governs the gap between bipartite correlation measured in trace norm and that measured by a correlation function, and obtain an improved universal upper bound for quantum data-hiding. The upper bound has also been formalized and machine-checked in Lean.
\end{abstract}
\section{Introduction and main results}\label{sec:introduction}

A bipartite quantum system may be probed by an observable acting jointly on the two subsystems or through observables applied separately to each subsystem. A joint observable has access to the full structure of the composite system, whereas a product observable only records the response to a pair of local probes. How much can be missed when only product observables are available?

This question can be naturally formulated as a comparison of different norms on the matrix space.
Let $M_k$ denote the algebra of complex $k\times k$ matrices, and identify $M_n\otimes M_m$ with $M_{nm}$. 
For $z\in M_n\otimes M_m$, the trace norm is obtained by testing $z$ against arbitrary contractive operators on the joint system:
\begin{equation}
  \norm{z}_1 = \sup_{\norm{y}_\infty\leq 1} \abs{\tr(yz)}.
\end{equation}
Restricting to product operators gives
\begin{equation}
  \norm{z}_\varepsilon = \sup_{\norm{a\otimes b}_\infty\le1} \abs{\tr\qty((a\otimes b)z)}. \label{eq:1.1}
\end{equation}
This is the injective tensor norm on $S_1^n\otimes S_1^m$, where $S_1^k$ is the space $M_k$ equipped with the trace norm. 
We seek the best ratio between these two norms:
\begin{equation}\label{eq:defC}
  C_{n,m}=\sup_{0\ne z\in M_n\otimes M_m} \frac{\norm{z}_1}{\norm{z}_\varepsilon}.
\end{equation}

Although the question is straightforward to pose, its answer is not immediate.
A priori, even the correct scaling with the local dimensions $n$ and $m$ is unclear.

To begin, note that any operator provides a lower bound
In particular, the SWAP operator $F_d$ on $\mathbb C^d\otimes\mathbb C^d$, embedded into $\mathbb C^n\otimes\mathbb C^m$, where $d=\min\{n,m\}$, gives:
\begin{equation}\label{eq:swaplowerbound}
  C_{n,m}\geq \min\{n,m\}.
\end{equation}
That is, $C_{n,m}$ must grow at least linearly with the smaller local dimension.

Turning to upper bounds, the story becomes richer. 
Brand{\~a}o and Horodecki first proved the direct estimate $C_{n,m}\leq \min\{n,m\}^2$ \cite[Appendix C]{BH15}.
Lami, Palazuelos, and Winter subsequently placed the problem in the framework of projective and injective tensor norms; their general comparison recovers the same quadratic bound \cite[Proposition 21]{LPW18}.
Aubrun, Lami, Palazuelos, Szarek, and Winter later improved the exponent, proving $C_{n,m}\lesssim \min\{n,m\}^{3/2}$ through a sharper estimate of the projective-to-injective norm ratio \cite[Theorem 8 and Corollary 9]{ALPSW20}.
Finally, using the noncommutative Grothendieck inequality \cite{Pis78,Haa85,Pis12}, Corr{\^e}a, Lami, and Palazuelos proved the direct linear estimate \cite[Proposition 2.3]{CLP22}
\begin{equation}
  C_{n,m} \leq 2\min\{n,m\},
\end{equation}
matching the lower bound \cref{eq:swaplowerbound} up to a constant factor 2, where 2 is the exact value of the noncommutative Grothendieck constant \cite{Haa85,HI95}.

At this point, it is tempting\footnote{The authors pursued this conjecture until finding a counterexample.} to conjecture that the SWAP lower bound is sharp, namely, that $C_{n,m}=\min\{n,m\}$. Two observations support this conjecture. First, the SWAP operator (relatedly, transpose map and maximally entangled states) is extremal in several norm-comparison problems in quantum information. 
Second, the commutative analogue\footnote{More precisely, for $a=(a_{ij})\in \ell_1^n\otimes\ell_1^m$, we have $\norm{a}_\varepsilon=\sup_{\abs{s_i},\abs{t_j}\leq1}\abs{\sum_{i,j}a_{ij}s_it_j}$ and $\norm{a}_1=\sum_{i,j}\abs{a_{ij}}$. 
Taking $s$ to be a coordinate vector and optimizing over $t$ gives $\norm{a}_\varepsilon\geq \max_i\sum_j\abs{a_{ij}}\geq n^{-1}\norm{a}_1$.} satisfies the bound $C_{n,m}^{\mathrm{(comm)}}\leq \min\{n,m\}$,
although a direct application of Grothendieck's inequality would instead give the weaker estimate $C_{n,m}^{\mathrm{(comm)}}\leq K_G^\mathbb{C}\min\{n,m\}$, where $K_C^\mathbb{C}> 1.338$ \cite{Pis12}.

It turns out this is not the case.
In this work, we improve the coefficient from $2$ to $\sqrt2$ and show that the latter is optimal.
\begin{theorem}\label{thm:main}
For every $n,m\ge1$ and every $z\in M_n\otimes M_m$,
\begin{equation}
  \norm{z}_1 \le \sqrt{2}\,\min\{n,m\}\,\norm{z}_\varepsilon. \label{eq:1.2}
\end{equation}
The constant is optimal in the sense that
\begin{equation}
  \lim_{n\to \infty}\sup_{m\ge1}\frac{C_{n,m}}{\min\{n,m\}} = \sqrt{2}. \label{eq:1.3}
\end{equation}
\end{theorem}

The main ingredient for the upper bound is a Haar-unitary $L_1$ noncommutative Khintchine inequality that we establish here, in the line of \cite{LP86,LP91,HM07} and formulated using the standard row-plus-column norm (see \cref{sec:preliminaries} below).
For every family $(x_{ij})_{i,j=1}^n\subset M_m$, we prove:
\begin{equation}
  \mathbb E_{w\in U(n)}\norm{\sum_{i,j=1}^n w_{ij}x_{ij}}_1
  \ge \frac1{\sqrt{2n}}\,\norm{(x_{ij})}_{L_1[R+C]}.
  \label{eq:1.4}
\end{equation}
To establish asymptotic sharpness, we utilize a fermionic canonical anticommutation relations (CAR)-based construction inspired by \cite{HI95,HM07}, and introduce a functional-calculus step that flattens the singular-value distribution.

We give two applications of our inequality.
For a bipartite state $\rho_{AB}$, the trace norm of the correlation operator $\rho_{AB}-\rho_A\otimes\rho_B$ measures the correlation detectable by arbitrary joint observables, whereas its injective norm is the largest correlation function detectable by product operators.
We show that the CAR-based construction yielding asymptotic sharpness for $C_{n,m}$ can be modified to produce correlation operators, so the bound $\sqrt{2}\min\{n,m\}$ remains asymptotically sharp under this restriction.
We then compare the injective norm with distinguishability norms arising from local measurements.
This yields the improved bound $R_{\mathrm{LO}}(n,m)\leq \sqrt{2}\min\{n,m\}$ for quantum data-hiding \cite{TDL01,DLT02,MWW09,LPW18,ALPSW20,CLP22}.

\section{Preliminaries and notation}\label{sec:preliminaries}

For $1\le p<\infty$, the Schatten $p$-norm on $M_k$ is $\norm{x}_p=\qty(\tr\abs{x}^p)^{1/p}$, while $\norm{\cdot}_\infty$ denotes the operator norm. All traces are unnormalized; when a normalized trace is used, it is denoted by $\tau$.

For $z=\sum_{i,j=1}^n E_{ij}\otimes x_{ij}\in M_n\otimes M_m$, define the associated linear map
\begin{equation}
  T_z:M_n\longrightarrow M_m, \qquad T_z(w) = \sum_{i,j=1}^n w_{ij}x_{ij}. \label{eq:2.1}
\end{equation}
Trace-norm duality, together with the fact that the operator-norm unit ball of $M_n$ is the closed convex hull of $U(n)$, gives
\begin{equation}
  \norm{z}_\varepsilon = \sup_{w\in U(n)} \norm{T_z(w)}_1. \label{eq:2.2}
\end{equation}

We also use the standard row and column norms from operator space theory \cite{Pis03}. For a finite family $x=(x_\alpha)\subset M_m$, define
\begin{equation}
\begin{aligned}
    \norm{x}_{L_1[C]}=\biggl\lVert(\sum_\alpha x_\alpha^*x_\alpha)^{1/2}\biggr\rVert_1, \qquad
    \norm{x}_{L_1[R]}=\biggl\lVert(\sum_\alpha x_\alpha x_\alpha^*)^{1/2}\biggr\rVert_1.
\end{aligned}
\end{equation}
The corresponding sum norm is
\begin{equation}
  \norm{x}_{L_1[R+C]} = \inf_{x_\alpha=a_\alpha+b_\alpha} \qty(\norm{a}_{L_1[R]} + \norm{b}_{L_1[C]}). \label{eq:2.3}
\end{equation}
Its dual is the intersection norm
\begin{equation}
  \norm{y}_{L_\infty[R\cap C]} = \max\qty{\biggl\lVert\sum_\alpha y_\alpha y_\alpha^*\biggr\rVert_\infty^{1/2}, \biggl\lVert\sum_\alpha y_\alpha^*y_\alpha\biggr\rVert_\infty^{1/2}}. \label{eq:2.4}
\end{equation}
Thus, $\norm{y}_{L_\infty[R\cap C]}\le1$ if and only if
\begin{equation}
  \sum_\alpha y_\alpha^*y_\alpha\le I, \qquad \sum_\alpha y_\alpha y_\alpha^*\le I. \label{eq:2.6}
\end{equation}
With respect to the trace pairing, sum-intersection duality gives
\begin{equation}
  \norm{x}_{L_1[R+C]} = \sup_{\norm{y}_{L_\infty[R\cap C]}\le1} \biggl\lvert\sum_\alpha\tr(y_\alpha^*x_\alpha)\biggr\rvert. \label{eq:2.5}
\end{equation}

\section{The upper bound}\label{sec:upper-bound}

In this section, we prove the upper bound \cref{eq:1.2}.
Throughout this section, we assume $n\le m$ without loss of generality; 
the general case follows by exchanging the two tensor factors.

The proof combines two estimates for block matrices.
The first, \cref{prop:block-matrix}, is an elementary estimate that relates the trace norm to the $L_1[R+C]$ norm.
The second, \cref{thm:khintchine}, is a sharp noncommutative Khintchine inequality proved by considering its dual statement, which takes the form of a bounded Fourier synthesis of matrix-valued functions over $U(n)$.

\subsection{Trace norm estimation for block matrices}
\begin{proposition}\label{prop:block-matrix}
For $z=\sum_{i,j=1}^n E_{ij}\otimes x_{ij}\in M_n\otimes M_m$, one has
\begin{equation}
  \norm{z}_1 \le \sqrt n\, \norm{(x_{ij})}_{L_1[R+C]}. \label{eq:3.1}
\end{equation}
\end{proposition}
\begin{proof}
We first prove the estimate for the $L_1[C]$ norm.
To do so, we decompose $z$ into columns.
Set $Z_j=\sum_{i=1}^n E_{ij}\otimes x_{ij}$; then $Z_j^*Z_j=E_{jj}\otimes\sum_i x_{ij}^*x_{ij}$, so $\norm{Z_j}_1=\tr(\sum_i x_{ij}^*x_{ij})^{1/2}$.
The triangle inequality and operator concavity of $t\mapsto t^{1/2}$ give
\begin{equation}
\begin{aligned}
    \norm{z}_1
    \le \sum_{j=1}^n\norm{Z_j}_1 
    =\sum_{j=1}^n \tr\Bigl(\sum_i x_{ij}^*x_{ij}\Bigr)^{1/2}
    \le \sqrt n\,\tr\Bigl(\sum_{i,j}x_{ij}^*x_{ij}\Bigr)^{1/2}
    =\sqrt n\,\norm{(x_{ij})}_{L_1[C]}.
\end{aligned}
  \label{eq:3.2}
\end{equation}

Applying the same argument to rows gives
\begin{equation}
  \norm{z}_1 \le \sqrt n\, \norm{(x_{ij})}_{L_1[R]}. \label{eq:3.3}
\end{equation}

For a decomposition $x_{ij}=a_{ij}+b_{ij}$, applying \cref{eq:3.3} to $(a_{ij})$ and \cref{eq:3.2} to $(b_{ij})$, we get:
\begin{equation}
\|z\|_1\le \|(a_{ij})\|_1+\|(b_{ij})\|_1 \le \sqrt n \|(a_{ij})\|_{L_1[R]} + \sqrt n \|(b_{ij})\|_{L_1[C]}.
\end{equation}
Taking the infimum over all decompositions proves \cref{eq:3.1}.
\end{proof}

\subsection{Noncommutative Khintchine inequality}

In this subsection, we prove the following version of the noncommutative Khintchine inequality.
\begin{theorem}\label{thm:khintchine}
For every family $(x_{ij})\subset M_m$ with $1\leq i,j\leq n$, one has:
\begin{equation}
  \mathbb E_{w\in U(n)} \norm{\sum_{i,j=1}^n w_{ij}x_{ij}}_1 \ge \frac1{\sqrt{2n}} \norm{(x_{ij})}_{L_1[R+C]}, \label{eq:3.29}
\end{equation}
Here the average is with respect to the standard Haar measure on $U(n)$.
\end{theorem}

We first introduce some terminology.
Given a function $F: U(n) \to M_m$, we may consider its first-order Fourier coefficients\footnote{We use this terminology in the spirit of Fourier analysis on the Lie group $U(n)$, since $w_{ij}$ are matrix coefficients of the fundamental (defining) representation of $U(n)$.}:
\begin{equation}
    y_{ij} = \mathbb E_w \qty[\overline{w_{ij}} F(w)] \label{eq:3.5}.
\end{equation}
Conversely, a family $(y_{ij})_{i,j=1}^n\subset M_m$ can always be produced in this way.
For example, we may define a first-order Fourier polynomial,
\begin{equation}
  S(w) = \sqrt n \sum_{i,j=1}^n w_{ij}y_{ij}, \qquad w\in U(n), \label{eq:3.4}
\end{equation}
and then Schur orthogonality gives
\begin{equation}
      y_{ij} = \mathbb E_w \qty[\overline{w_{ij}} \sqrt n S(w)].
\end{equation}
However, such an $F(w)$ is not unique, as we may also include higher-order Fourier terms from other representations of $U(n)$:
\begin{equation}
  F(w) = \sqrt n S(w) + (\text{higher-order terms}).
\end{equation}

The key claim in this subsection is that, given $(y_{ij})_{i,j=1}^n\subset M_m$, there always exists an $F(w)$ with moderate $L_\infty(U(n);M_m)$ norm compared with $\norm{(y_{ij})}_{L_\infty[R\cap C]}$. 
More precisely, we claim:

\begin{theorem}\label{thm:fourier-interpolation}
  Let $(y_{ij})_{i,j=1}^n\subset M_m$ satisfy $ \|(y_{ij})\|_{L_\infty[R\cap C]}\le1 $.
There exists a measurable function $F:U(n)\to M_m$ such that
\begin{align}
  &\mathbb E_w \qty[\overline{w_{ij}}F(w)] = y_{ij}
  \qquad (1\le i,j\le n), \label{eq:3.7}\\
  &\norm{F(w)}_\infty 
  \le \sqrt{2n} \quad(\forall w\in U(n)). \label{eq:3.8}
\end{align}
\end{theorem}

The proof is inspired by the truncation and iteration used by Haagerup and Musat in the proof of their own noncommutative Khintchine inequality \cite{HM07}. 
The main idea is that, although $\sqrt n S(w)$ may have a large $L_\infty(U(n);M_m)$ norm, we can truncate it, project the remainder onto the first-order Fourier space, and iterate.
Throughout the iteration, the function may not be a first-order Fourier polynomial anymore, yet the first-order Fourier coefficients are preserved.

For $\tau>0$, we define the singular-value truncation as follows.
Given $X\in M_m$, we take a singular value decomposition $X=\sum_r s_r u_rv_r^*$ and set
\begin{equation}
  \operatorname{trunc}_\tau(X)=\sum_r\min\{s_r,\tau\}u_rv_r^*,\qquad
  \rho_\tau(X)=X-\operatorname{trunc}_\tau(X)=\sum_r(s_r-\tau)_+u_rv_r^*. \label{eq:3.17}
\end{equation}
Thus, $\operatorname{trunc}_\tau$ caps each singular value at $\tau$, while $\rho_\tau$ records the excess above $\tau$.
Both maps are well-defined even if the singular-value decomposition is not unique.

\begin{proposition}\label{prop:fourth-moment}
  Let $(y_{ij})_{i,j=1}^n\subset M_m$ satisfy $ \|(y_{ij})\|_{L_\infty[R\cap C]}\le1 $.
Define $S(w)$ by \cref{eq:3.4}.
Then
\begin{equation}
  \mathbb E_w(S^*S)^2 \le 2I, \qquad \mathbb E_w(SS^*)^2 \le 2I. \label{eq:3.9}
\end{equation}
\end{proposition}
\begin{proof}
  We may assume $n\ge 2$ since the $n=1$ case is trivial. 
It suffices to prove the first inequality. The second follows by applying the same argument to the adjoint family $(y_{ij}^*)$.

Using \cref{eq:3.4}, we expand $(S^*S)^2$ as:
\begin{equation}
  (S^*S)^2
  =
  n^2\sum_{i,j,k,l,p,q,r,s}
  \overline{w_{ij}}w_{kl}\overline{w_{pq}}w_{rs}
  y_{ij}^*y_{kl}y_{pq}^*y_{rs}.
\end{equation}
We use the Haar moment formula \cite{Collins03,CS06}
\begin{equation}
\begin{aligned}
    \mathbb E\qty[\overline{w_{ij}}w_{kl}\overline{w_{pq}}w_{rs}]
    =
    \frac{\delta_{ki}\delta_{rp}\delta_{lj}\delta_{sq}
    +\delta_{kp}\delta_{ri}\delta_{lq}\delta_{sj}}{n^2-1} 
    -
    \frac{\delta_{ki}\delta_{rp}\delta_{lq}\delta_{sj}
    +\delta_{kp}\delta_{ri}\delta_{lj}\delta_{sq}}{n(n^2-1)}.
\end{aligned}
  \label{eq:3.10}
\end{equation}
This gives:
\begin{equation}
  \mathbb E(S^*S)^2 = \frac{n^2}{n^2-1}(A^2+P) - \frac{n}{n^2-1}(N_1+N_2). \label{eq:3.11}
\end{equation}
Here,
\begin{align}
  A=\sum_{i,j}y_{ij}^*y_{ij}&,\qquad
  B=\sum_{i,j}y_{ij}y_{ij}^*,\qquad
  P=\sum_{i,j}y_{ij}^*By_{ij},\\
  N_1 &= \sum_{j,q} \qty(\sum_i y_{ij}^*y_{iq})
  \qty(\sum_p y_{pq}^*y_{pj}), \label{eq:3.12}\\
  N_2 &= \sum_{i,p} \qty(\sum_j y_{ij}^*y_{pj})
  \qty(\sum_q y_{pq}^*y_{iq}). \label{eq:3.13}
\end{align}
Recall \cref{eq:2.6}, so we have $A\le I$ and $B\le I$.
This implies $A^2\le I$ and 
\begin{equation}
  P \le \sum_{i,j}y_{ij}^*y_{ij} = A \le I. \label{eq:3.14}
\end{equation}

We next show that $N_1$ and $N_2$ are positive and bounded below by $\frac1nA^2$.
Define $C_{jq}=\sum_i y_{ij}^*y_{iq}$. 
Since $C_{qj}=C_{jq}^*$, we have
\begin{equation}
  N_1=\sum_{j,q}C_{jq}C_{jq}^*\ge0.
\end{equation}
Keeping only the terms with $j=q$ and using operator convexity of $t\mapsto t^2$, we obtain
\begin{equation}
  N_1\ge\sum_jC_{jj}^2\ge\frac1n\qty(\sum_jC_{jj})^2=\frac1nA^2. \label{eq:3.15}
\end{equation}
The same argument gives
\begin{equation}
  N_2\ge\frac1nA^2. \label{eq:3.16}
\end{equation}

Substituting \crefrange{eq:3.14}{eq:3.16} into \cref{eq:3.11}, we obtain
\begin{equation}
  \mathbb E(S^*S)^2
  \le \frac{(n^2-2)A^2+n^2P}{n^2-1}
  \le \frac{2n^2-2}{n^2-1}I
  =2I.
  \label{eq:fourth-moment-bound}
\end{equation}
This proves the proposition.
\end{proof}

\begin{proposition}\label{prop:truncation}
  Let $(y_{ij})_{i,j=1}^n\subset M_m$ satisfy $ \|(y_{ij})\|_{L_\infty[R\cap C]}\le1 $.
Define $S(w)$ by \cref{eq:3.4}, and define
\begin{equation}
  y_{ij}^{(1)} = \sqrt n\, \mathbb E_w \qty[\overline{w_{ij}} \rho_{\frac{1}{\sqrt2}}(S(w))]. \label{eq:3.18}
\end{equation}
Then
\begin{equation}
  \norm{(y_{ij}^{(1)})}_{L_\infty[R\cap C]} \le \frac12. \label{eq:3.19}
\end{equation}
\end{proposition}
\begin{proof}
For $s\ge0$, one has $\qty(s-\frac1{\sqrt2})_+\le \frac{s^2}{2\sqrt2}$. 
Functional calculus and \Cref{prop:fourth-moment} therefore give
\begin{equation}
\begin{aligned}
    \mathbb E_w\,\rho_{\frac{1}{\sqrt2}}(S(w))^*\rho_{\frac{1}{\sqrt2}}(S(w))
    =\mathbb E_w\qty(\abs{S(w)}-\frac1{\sqrt2}I)_+^2 
    \le \frac18\,\mathbb E_w\qty(S(w)^*S(w))^2\le\frac14I.
\end{aligned}
  \label{eq:3.20}
\end{equation}
The functions $\sqrt n\,w_{ij}$ are orthonormal in $L_2(U(n))$. Bessel's inequality yields
\begin{equation}
  \sum_{i,j}(y_{ij}^{(1)})^*y_{ij}^{(1)} \le \mathbb E_w\, \rho_{\frac{1}{\sqrt2}}(S(w))^* \rho_{\frac{1}{\sqrt2}}(S(w)) \le \frac14I. \label{eq:3.21}
\end{equation}
Applying the same argument to adjoints gives
\begin{equation}
  \sum_{i,j}y_{ij}^{(1)}(y_{ij}^{(1)})^* \le \mathbb E_w\,\rho_{\frac{1}{\sqrt2}}(S(w))\rho_{\frac{1}{\sqrt2}}(S(w))^* \le \frac14I. \label{eq:3.22}
\end{equation}
\Cref{eq:3.21,eq:3.22} prove \cref{eq:3.19}.
\end{proof}

\begin{proof}[Proof of \Cref{thm:fourier-interpolation}]

  We will iteratively define two families of functions $S_k: U(n)\to M_m$ and $F_k: U(n)\to M_m$, indexed by $k\in\mathbb N$. 

To start, we define $S_0(w)=\sqrt n\sum_{i,j}w_{ij}y_{ij}$.
At each step, once we have $S_k$, we define
\begin{equation}\label{eq:ykfromSk}
  y_{ij}^{(k)}=\sqrt n\,\mathbb E_w\qty[\overline{w_{ij}}S_k(w)],
  \qquad
  r_k=\norm{(y_{ij}^{(k)})}_{L_\infty[R\cap C]}.
\end{equation}
If $r_k=0$, set all later terms equal to zero.
Otherwise, decompose $S_k(w)$ pointwise as
\begin{equation}\label{eq:decomposeSk}
  S_k(w)
  =
  \operatorname{trunc}_{\frac{r_k}{\sqrt2}}(S_k(w))
  +
  \rho_{\frac{r_k}{\sqrt2}}(S_k(w)).
\end{equation}
We keep the truncated part to define $F_k$, and use the first-order coefficients of the excess part to define the next polynomial $S_{k+1}$:
\begin{equation}
  F_k(w) = \sqrt n\, \operatorname{trunc}_{\frac{r_k}{\sqrt2}}(S_k(w)), \label{eq:3.23}
\end{equation}
\begin{equation}
\begin{aligned}
  y_{ij}^{(k+1)}
  =
  \sqrt n\,\mathbb E_w\qty[\overline{w_{ij}}\rho_{\frac{r_k}{\sqrt2}}(S_k(w))], \qquad
  S_{k+1}(w)
  =
  \sqrt n\sum_{i,j}w_{ij}y_{ij}^{(k+1)}.
\end{aligned}
  \label{eq:3.24}
\end{equation}

Once we have all the $S_k$ and $F_k$, our construction of $F(w)$ will be
\begin{equation}
  F(w)=\sum_{k=0}^\infty F_k(w).
\end{equation}
In the following, we prove that it converges and satisfies the claims of \cref{thm:fourier-interpolation}.

Schur orthogonality gives $y_{ij}^{(0)}=y_{ij}$, so $r_0\le1$.
By homogeneity, \cref{prop:truncation} applied to $(y_{ij}^{(k)}/r_k)$ gives
\begin{equation}
  r_{k+1} \le \frac12r_k, \label{eq:3.26}
\end{equation}
hence $r_k\le2^{-k}$.
Moreover, by definition of the truncation, we have:
\begin{equation}
  \norm{F_k(w)}_\infty \le \sqrt{\frac n2}\,r_k \qquad \forall w\in U(n). \label{eq:3.27}
\end{equation}
Therefore, $F=\sum_{k=0}^\infty F_k$ converges uniformly and satisfies
\begin{equation}
  \norm{F}_\infty
  \le
  \sum_{k=0}^\infty\norm{F_k}_\infty
  \le
  \sqrt{\frac n2}\sum_{k=0}^\infty2^{-k}
  =
  \sqrt{2n}. \label{eq:3.28}
\end{equation}

The definitions \cref{eq:ykfromSk,eq:decomposeSk,eq:3.23,eq:3.24} imply that the first-order coefficients satisfy:
\begin{equation}
  y_{ij}^{(k)} = \mathbb E_w \qty[\overline{w_{ij}}F_k(w)] + y_{ij}^{(k+1)}. \label{eq:3.25}
\end{equation}
Iterating this equation over $k$ and using that $\norm{y_{ij}^{(k)}}_\infty\leq  r_{k}\to 0$ and $\sum_kF_k$ converges uniformly, we obtain
\begin{equation}
  y_{ij}
  =
  \sum_{k=0}^\infty\mathbb E_w\qty[\overline{w_{ij}}F_k(w)]
  =
  \mathbb E_w\qty[\overline{w_{ij}}F(w)].
\end{equation}
This proves \cref{eq:3.7,eq:3.8}.
\end{proof}

The dual form of \cref{thm:fourier-interpolation} proves the desired \cref{thm:khintchine}.

\begin{proof}[Proof of \cref{thm:khintchine}]
Fix $(y_{ij})$ in the unit ball of $L_\infty[R\cap C]$. By \cref{thm:fourier-interpolation}, there exists $F:U(n)\to M_m$ such that $\mathbb E_w[\overline{w_{ij}}F(w)]=y_{ij}$ and $\norm{F}_\infty\le\sqrt{2n}$. Therefore,
\begin{equation}
\begin{aligned}
    \biggl\lvert\sum_{i,j}\tr(y_{ij}^*x_{ij})\biggr\rvert
    =\biggl\lvert\mathbb E_w\tr\Bigl(F(w)^*\sum_{i,j}w_{ij}x_{ij}\Bigr)\biggr\rvert
    \le\sqrt{2n}\,\mathbb E_w\biggl\lVert\sum_{i,j}w_{ij}x_{ij}\biggr\rVert_1.
\end{aligned}
  \label{eq:khintchine-duality}
\end{equation}
Taking the supremum over $(y_{ij})$ and using the duality \cref{eq:2.5} proves the theorem.
\end{proof}

For completeness, we also give the reverse estimate.
\begin{proposition}
  For every family $(x_{ij})\subset M_m$ with $1\leq i,j\leq n$, one has:
\begin{equation}
  \mathbb E_{w\in U(n)}\norm{\sum_{i,j=1}^n w_{ij}x_{ij}}_1
  \leq \frac{1}{\sqrt n}\,\norm{(x_{ij})}_{L_1[R+C]}.
  \label{eq:reverse-khintchine}
\end{equation}
This estimate is sharp.
\end{proposition}

\begin{proof}
Fix a decomposition $x_{ij}=a_{ij}+b_{ij}$ and write $A(w)=\sum_{i,j}w_{ij}a_{ij}$ and $B(w)=\sum_{i,j}w_{ij}b_{ij}$. By Jensen's inequality for the operator-concave function $t\mapsto t^{1/2}$ and Haar orthogonality,
\begin{align}
  \mathbb E_w\norm{A(w)}_1
  &\leq \tr\qty(\mathbb E_w A(w)A(w)^*)^{1/2}
  =\frac{1}{\sqrt n}\,\norm{(a_{ij})}_{L_1[R]}, \\
  \mathbb E_w\norm{B(w)}_1
  &\leq \tr\qty(\mathbb E_w B(w)^*B(w))^{1/2}
  =\frac{1}{\sqrt n}\,\norm{(b_{ij})}_{L_1[C]}.
\end{align}
The triangle inequality and the infimum over all such decompositions prove \cref{eq:reverse-khintchine}.

To see that the bound is sharp, take $m=n$ and $x_{ij}=E_{ij}$. Then $\sum_{i,j}w_{ij}x_{ij}=w$, so the l.h.s of \cref{eq:reverse-khintchine} is $n$. On the other hand, 
\begin{equation}
  \norm{(E_{ij})}_{L_1[R+C]}
  \leq \norm{(E_{ij})}_{L_1[R]}
  =\biggl\lVert\biggl(\sum_{i,j}E_{ij}E_{ij}^*\biggr)^{1/2}\biggr\rVert_1
  =n\sqrt n.
\end{equation}
This shows that the bound \cref{eq:reverse-khintchine} is sharp.
\end{proof}

\subsection{Final upper bound and Lean formalization}
Combining \cref{prop:block-matrix,thm:khintchine}, we obtain
\begin{equation}
\begin{aligned}
    \norm{z}_1
    \le\sqrt n\,\norm{(x_{ij})}_{L_1[R+C]} 
    \le\sqrt2\,n\,\mathbb E_w\norm{T_z(w)}_1 
    \le\sqrt2\,n\,\norm{z}_\varepsilon.
\end{aligned}
  \label{eq:3.30}
\end{equation}
This concludes the first part of the main result, \cref{thm:main} (recall that we have assumed $n\le m$ without loss of generality).

The complete upper-bound argument has also been formalized and machine-checked in Lean 4. The source code is available in the accompanying repository \cite{LeanFormalization}. 
The upper bound in \cref{thm:main} corresponds to the declaration \texttt{UpperBound.upper\_bound}. The formalization includes the block-matrix estimate \cref{prop:block-matrix}, the noncommutative Khintchine inequality \cref{thm:khintchine}, and their combination in \cref{eq:3.30}.

\section{The lower bound}\label{sec:lower-bound}

In this section, we prove that the constant $\sqrt 2$ is optimal, in the sense that there exists a family $z^{(n)}\in M_n\otimes M_{D_n}$ ($D_n\geq n$), such that:
\begin{equation}\label{eq:saturate}
  \lim_{n\to\infty} \frac{\norm{z^{(n)}}_1}{n\norm{z^{(n)}}_\varepsilon} = \sqrt{2}.
\end{equation}

\subsection{The CAR construction}

Our construction is based on the canonical anticommutation relations (CAR), which have proved useful in constructing extremal examples in operator space theory. 
Notably, Haagerup and Itoh used a CAR construction to establish the sharpness of the constant \(2\) in the noncommutative Grothendieck inequality \cite{HI95}; related CAR constructions underlie Haagerup and Musat's determination of sharp constants in their noncommutative Khintchine-type inequalities \cite{HM07}.
Our additional step is to pass to the partial-isometry part of the CAR matrix, thereby flattening all of its nonzero singular values.

Let $\mathcal F_n=\Lambda(\mathbb C^{n^2})$ be the Fock space of $n^2$ complex fermionic modes, indexed by pairs $(i,j)$ with $1\le i,j\le n$. Its dimension is $D_n=2^{n^2}$.
Denote by $c_{ij}$ the annihilation operator associated with the mode $(i,j)$. 
The canonical anticommutation relations are
\begin{equation}
  \acomm{c_{ia}}{c_{kb}}=0, \qquad \acomm{c_{ia}}{c_{kb}^*} = \delta_{ik}\delta_{ab}I. \label{eq:4.1}
\end{equation}

Define (we omit the $n$-dependence in the notation)
\begin{align}
  C = \sum_{i,j=1}^n E_{ij}\otimes c_{ij}
  \in M_n\otimes M_{D_n}, \label{eq:4.3}
\end{align}
and polar decompose it as
\begin{equation}
  C=V\abs{C},
  \label{eq:polar-decomposition}
\end{equation}
such that $V$ is the unique partial isometry satisfying \cref{eq:polar-decomposition} and $\ker V=\ker C$.
$V$ is a singular-value-flattened version of $C$ in which all nonzero singular values are set to 1.

We will use $V$ as a witness to certify \cref{eq:saturate}.

\subsection{Injective tensor norm estimate}

We next estimate the injective norm of $V$.
Write
\begin{equation}
  V=\sum_{i,j=1}^n E_{ij}\otimes x_{ij},\qquad T_{V}(a)=\sum_{i,j=1}^na_{ij} x_{ij},\qquad (a\in M_n).
\end{equation}
By \cref{eq:2.2}, it suffices to control $\norm{T_V(w)}_1$ for $w\in U(n)$.

Suppose we do not perform polar decomposition and consider $T_C(a)$ instead. Then it is a linear combination of annihilation operators and is therefore square-zero: $T_C(a)^2=0$. 
The next proposition shows that this still holds for $T_V(a)$.
In fact, we prove a stronger statement.

\begin{proposition}\label{prop:square-zero}
  For any one-variable polynomial $q$, define $C_q=C q(C^*C)$.  
  Define $T_q: M_n \to M_{D_n}$ by
\begin{equation}
  T_q(a) = \sum_{i,j=1}^n a_{ij} y_{ij},\qquad C_q=\sum_{i,j=1}^n E_{ij}\otimes y_{ij}.
\end{equation}
Then, for every $a\in M_n$,
\begin{equation}
  T_q(a)^2 = 0. \label{eq:4.5}
\end{equation}
\end{proposition}
The proof is an elementary calculation based on the CAR, and is deferred to the \hyperref[app:proof-square-zero]{appendix}.

In particular, we can pick $q$ to be a polynomial (depending on $n$) such that 
\begin{equation}
\qquad q(\lambda)=\frac{1}{\sqrt{\lambda}}~~~~(\forall \lambda\in \operatorname{spec}(C^*C)\setminus\{0\}), 
\end{equation}
For this choice, we have $V=Cq(C^*C)$. Therefore,
\begin{equation}\label{eq:vn-square-zero}
  T_V(a)^2 = 0.
\end{equation}

The square-zero identity gives the required injective-norm estimate.
\begin{corollary}\label{cor:injective-estimate}
For every $a\in M_n$,
\begin{equation}\label{eq:4.11}
  \norm{T_V(a)}_1\le \sqrt{\frac{D_n}{2}}\,\frac{\norm{V}_2}{n}\norm{a}_2.
\end{equation}
Consequently,
\begin{equation}\label{eq:Vinj}
  \norm{V}_\varepsilon\le \sqrt{\frac{D_n}{2n}}\,\norm{V}_2.
\end{equation}
\end{corollary}

\begin{proof}
From $T_V(a)^2=0$, one deduces
\begin{equation}
  \rank T_V(a)\le \frac{D_n}{2}. \label{eq:4.12}
\end{equation}
(Notice that $\operatorname{Im}T_V(a)\subseteq\ker T_V(a)$, hence $  \rank T_V(a)  \le   \dim\ker T_V(a)   =   D_n-\rank T_V(a)$.)
It follows that
\begin{align}
  \norm{T_V(a)}_1&\le \sqrt{\frac{D_n}{2}}\,\norm{T_V(a)}_2. \label{eq:4.13}
\end{align}

For $u,v\in U(n)$, consider the action
\begin{equation}
  c_{ij} \longmapsto \sum_{k,l}u_{ik}c_{kl}v_{lj}. \label{eq:4.14}
\end{equation}
This action preserves the CAR, hence is implemented unitarily on $\mathcal F_n$; $C$ and $V$ also transform covariantly under it.
Consequently, the quadratic form $a\mapsto\norm{T_V(a)}_2^2$ is invariant under the left-right action of $U(n)\times U(n)$ on $M_n$.
Such a quadratic form is unique up to a constant (due to the fact that this representation is irreducible and Schur's lemma), hence $\norm{T_V(a)}_2=\lambda\norm{a}_2$ for some $\lambda\ge 0$.

To determine $\lambda$, we take $a$ to be each elementary matrix $E_{ij}$ and sum over the squared identities, which gives
\begin{equation}
  n^2\lambda^2
  =
  \sum_{i,j}\norm{T_V(E_{ij})}_2^2
  =
  \sum_{i,j}\norm{x_{ij}}_2^2
  =
  \norm{V}_2^2.
\end{equation}
Hence, for every $a\in M_n$,
\begin{equation}
  \norm{T_V(a)}_2
  =\frac{\norm{V}_2}{n}\norm{a}_2.
  \label{eq:4.15}
\end{equation}
Combining \cref{eq:4.13,eq:4.15} proves the first estimate \cref{eq:4.11}. 

\Cref{eq:Vinj} follows by taking $a\in U(n)$, for which $\norm{a}_2=\sqrt n$, and using \cref{eq:2.2}.
\end{proof}

\subsection{The quarter-circle law}

Define $\mathrm{rk}_n=\rank V=\rank C$.
Since $V$ is a partial isometry, we have
\begin{equation}
  \norm{V}_1=\mathrm{rk}_n,
  \qquad
  \norm{V}_2=\sqrt{\mathrm{rk}_n}.
  \label{eq:V-Schatten-norms}
\end{equation}
It follows from \cref{cor:injective-estimate} that
\begin{equation}
  \frac{\norm{V}_1}{n\norm{V}_\varepsilon}
  \ge \sqrt{\frac2{nD_n}}\,\frac{\norm{V}_1}{\norm{V}_2}
  =\sqrt{\frac{2\mathrm{rk}_n}{nD_n}}.
  \label{eq:4.27}
\end{equation}
Thus it remains to show that $\mathrm{rk}_n/(nD_n)\to1$.

We prove the stronger statement about the limiting behaviour of the singular-value distribution of the CAR matrix $C$.
Let $\sigma_1,\ldots,\sigma_{nD_n}$ be the singular values of $\sqrt{\frac{2}{n}}C$, counted with multiplicity, and define the empirical distribution
\begin{equation}
  \nu_n := \frac1{nD_n} \sum_{k=1}^{nD_n} \delta_{\sigma_k}. \label{eq:4.16}
\end{equation}
We will prove that its limiting distribution as $n\to\infty$ is the quarter-circle law:
\begin{equation}
  d\nu_{\mathrm{QC}}(\sigma) = \frac1\pi \sqrt{4-\sigma^2}\, \mathbf1_{[0,2]}(\sigma)\,d\sigma. \label{eq:4.17}
\end{equation}

\begin{theorem}\label{thm:quarter-circle}
As $n\to\infty$, we have the weak convergence:
\begin{equation}
  \nu_n \Longrightarrow \nu_{\mathrm{QC}}. \label{eq:4.18}
\end{equation}
\end{theorem}

Although the CAR matrix has no randomness, our proof follows the standard method in random matrix theory: compute the moments combinatorially and match them with those of the limit distribution \cite{MP67,Bia97,MS17}.
The proof is deferred to the \hyperref[app:proof-quarter-circle]{appendix}.

\begin{corollary}\label{cor:asymptotic-rank}
As $n\to\infty$, one has
\begin{equation}
  \frac{\mathrm{rk}_n}{nD_n}\longrightarrow1.
\end{equation}
\end{corollary}

\begin{proof}
  Since $\{0\}$ is a closed set, the weak convergence gives 
\begin{equation}
  \limsup_{n\to\infty}\nu_n(\{0\})
  \le \nu_{\mathrm{QC}}(\{0\})
  =0.
\end{equation}
On the other hand,
\begin{equation}
  \nu_n(\{0\})
  =\frac{\dim\ker C}{nD_n}
  =1-\frac{\mathrm{rk}_n}{nD_n}.
\end{equation}
This proves the claim.
\end{proof}

Combining \cref{cor:asymptotic-rank} with \cref{eq:4.27}, we obtain
\begin{equation}
  \lim_{n\to\infty}\frac{\norm{V}_1}{n\norm{V}_\varepsilon} \ge \sqrt2.
  \label{eq:4.28}
\end{equation}
Taking $z^{(n)}=V$ and combining this with the upper bound in \cref{thm:main} proves \cref{eq:saturate} and completes the proof of \cref{thm:main}.

\section{Applications}\label{sec:applications}
\subsection{Measures of bipartite correlation}

A bipartite state is uncorrelated precisely when it is the product of its marginals. 
Its departure from this product state may be tested either by product operators on the two subsystems or by an arbitrary operator on the joint system. 
The former gives the usual correlation function, while the latter measures the same departure in trace norm.

More precisely, let $\rho_{AB}$ be a state on $\mathbb C^n\otimes\mathbb C^m$, with marginals $\rho_A=\tr_B\rho_{AB}$ and $\rho_B=\tr_A\rho_{AB}$. 
Define the correlation operator by
\begin{equation}
  \Delta_\rho = \rho_{AB}-\rho_A\otimes\rho_B. \label{eq:5.1}
\end{equation}
For two operators $a\in M_n$ and $b\in M_m$, the corresponding correlation function is
\begin{equation}
  \operatorname{Cor}_\rho(a,b)
  = \tr\qty((a\otimes b)\rho_{AB})-\tr(a\rho_A)\tr(b\rho_B)
  = \tr\qty((a\otimes b)\Delta_\rho). \label{eq:5.2}
\end{equation}
So we have:
\begin{align}
  \mathcal C_{\mathrm{prod}}(\rho)
  &= \sup_{\substack{\norm{a}_\infty\le1\\ \norm{b}_\infty\le1}}
  \abs{\operatorname{Cor}_\rho(a,b)}
  = \norm{\Delta_\rho}_\varepsilon, \label{eq:5.3}\\
  \mathcal C_{\mathrm{glob}}(\rho)
  &= \sup_{\norm{X}_\infty\le1}\abs{\tr(X\Delta_\rho)}
  = \norm{\Delta_\rho}_1. \label{eq:5.4}
\end{align}
\Cref{thm:main} gives
\begin{equation}
  \mathcal C_{\mathrm{glob}}(\rho)
  \le \sqrt2\,\min\{n,m\}\,\mathcal C_{\mathrm{prod}}(\rho). \label{eq:5.5}
\end{equation}
However, it is not a priori clear whether $\sqrt{2}$ remains optimal under this restriction, since the witness $V$ from \cref{sec:lower-bound} is non-Hermitian and therefore cannot itself be a correlation operator.
We show that the same constant is nevertheless optimal by modifying this witness appropriately.

We begin by characterizing the operators that arise, up to scalar multiplication, as correlation operators.
\begin{proposition}\label{prop:correlation-operators}
A nonzero $h\in M_n\otimes M_m$ is a scalar multiple of $\Delta_\rho$ for some bipartite state $\rho_{AB}$ if and only if
\begin{equation}
  h=h^*, \qquad \tr_A h=0, \qquad \tr_B h=0. \label{eq:5.6}
\end{equation}
\end{proposition}
\begin{proof}
The necessity of \cref{eq:5.6} is obvious. 
Conversely, suppose that they hold. We consider the operator
\begin{equation}
  \rho_{AB}
  = \frac{I_n}{n}\otimes\frac{I_m}{m}+\delta h, \label{eq:5.7}
\end{equation}
where $\delta\in \mathbb R$.
When $|\delta|\le1/(nm\norm{h}_\infty)$, $\rho_{AB}$ is positive semidefinite, has unit trace, and hence is a state. 
Its marginals are $I_n/n$ and $I_m/m$, so $\Delta_\rho=\delta h$.
\end{proof}

Since the ratio of the two norms is homogeneous, \cref{prop:correlation-operators} reduces the optimization over states to Hermitian operators with vanishing partial traces.

\begin{theorem}\label{thm:correlation}
\begin{equation}
\sup_{\substack{n,m\\
\rho_{AB}\ne\rho_A\otimes\rho_B}}
  \frac{\mathcal C_{\mathrm{glob}}(\rho)}
  {\min\{n,m\} \mathcal C_{\mathrm{prod}}(\rho)}
  = \sqrt2. \label{eq:5.8}
\end{equation}
\end{theorem}

\begin{proof}
The upper bound follows from \cref{eq:5.5}. For the lower bound, we will modify the CAR witness $V$ from \cref{sec:lower-bound}. 

To start, we define
\begin{equation}
    S=\tr_AV=T_V(I_n), \qquad
  \widetilde V=V-\frac{I_n}{n}\otimes S.
  \label{eq:centered-car-witness}
\end{equation}
Then $\tr_A\widetilde V=0$. 
Now we estimate the norms of $\widetilde V$.
Denote $b=\sqrt{\frac{D_n}{2n}}\norm{V}_2$ for convenience.
First, 
\begin{equation}
  \lVert\widetilde V\rVert_1
  \ge
  \norm{V}_1-\norm{S}_1,
  \qquad
\norm{S}_1\le b
\end{equation}
Here the second step is due to \cref{cor:injective-estimate}.
Second, for $w\in U(n)$, we have
\begin{align}
  &T_{\widetilde V}(w)
  =T_V\qty(w-\frac{\tr(w)}nI_n),\\
  \norm{T_{\widetilde V}(w)}_1
  \le
 \sqrt{\frac{D_n}{2}}\frac{\norm{V}_2}{n}& \norm{w-\frac{\tr(w)}nI_n}_2
 =\sqrt{\frac{D_n}{2}}\frac{\norm{V}_2}{n}\Bigl(n-\frac{\abs{\tr(w)}^2}{n}\Bigr)^{\frac12}
\leq b.
\end{align}
Hence $\lVert\widetilde V\rVert_\varepsilon\le b$ by \cref{eq:2.2}. 

It follows that
\begin{equation}\label{eq:wildezqratio}
  \frac{\lVert\widetilde V\rVert_1}{n\lVert\widetilde V\rVert_\varepsilon}
  \ge\frac{\norm{V}_1-b}{nb}
  = \sqrt{\frac{2}{nD_n}}\frac{\norm{V}_1}{\norm{V}_2} -\frac1n.  
\end{equation}
We have proved in \cref{cor:asymptotic-rank} that the first term converges to $\sqrt2$, while the second term $1/n$ tends to zero. 
Thus the centered witness $\widetilde V$ also attains the same asymptotic constant.

Next, we ``hermitize" it. Define
\begin{equation}
  h
  =\widetilde V\otimes E_{12}
  +\widetilde V^*\otimes E_{21}
  \in M_n\otimes M_{2D_n}.
\end{equation}
Then $h=h^*$ and $\tr_Ah=0$. Moreover, $\tr_Bh=0$. 
We estimate the norms of $h$.
The trace norm is simple:
\begin{equation}
  \norm{h}_1=2\lVert\widetilde V\rVert_1.
\end{equation}
For the injective norm, writing the test operator in $M_{2D_n}$ as a $2\times2$ block matrix, we are left with two resulting pairings, each bounded by $\lVert\widetilde V\rVert_\varepsilon$. Thus
\begin{equation}
  \norm{h}_\varepsilon\le2\lVert\widetilde V\rVert_\varepsilon.
\end{equation}
Therefore,
\begin{equation}
  \frac{\norm{h}_1}{n\norm{h}_\varepsilon}
  \ge\frac{\lVert\widetilde V\rVert_1}{n\lVert\widetilde V\rVert_\varepsilon}.
\end{equation}
Together with \cref{eq:wildezqratio}, this implies:
\begin{equation}
  \sup_{n,m}\frac1{\min\{n,m\}}
  \sup_{\substack{0\ne h=h^*\\ \tr_Ah=0\\ \tr_Bh=0}}
  \frac{\norm{h}_1}{\norm{h}_\varepsilon}
  =\sqrt2. \label{eq:5.9}
\end{equation}
This proves \cref{eq:5.8}.
\end{proof}

\subsection{Quantum data hiding}

Quantum data hiding is the phenomenon that two bipartite states may be readily distinguished by a joint measurement yet remain difficult to distinguish when the two parties are restricted to local measurements (and classical communication) \cite{TDL01,DLT02}.
This leads naturally to a quantitative question: how much distinguishing power can be lost when arbitrary joint measurements are replaced by a restricted family \cite{LPW18,CLP22}?

Following the distinguishability-norm framework of \cite{MWW09}, let $\mathcal M$ be a class of POVMs.
For a Hermitian operator $h$, define
\begin{equation}
  \norm{h}_{\mathcal M}
  =
  \sup_{(M_\gamma)\in\mathcal M}
  \sum_\gamma\abs{\tr(M_\gamma h)}.
  \label{eq:measurement-norm}
\end{equation}
Suppose that the system is prepared in a state $\rho$ with probability $p$ or in a state $\sigma$ with probability $1-p$, where $0\le p\le1$, and set $h=p\rho-(1-p)\sigma$.
The optimal bias under arbitrary joint measurements is $\norm{h}_1$, while the largest bias achievable with measurements from $\mathcal M$ is $\norm{h}_{\mathcal M}$.
For fixed local dimensions $n,m\geq2$, the worst-case loss is quantified by the data-hiding ratio
\begin{equation}
\begin{aligned}
    R_{\mathcal M}(n,m)
    =
    \sup_{\substack{p,\rho,\sigma\\ p\rho\ne(1-p)\sigma}}
    \frac{\norm{p\rho-(1-p)\sigma}_1}
         {\norm{p\rho-(1-p)\sigma}_{\mathcal M}} 
    =
    \sup_{0\ne h=h^*}
    \frac{\norm{h}_1}{\norm{h}_{\mathcal M}}.
\end{aligned}
  \label{eq:data-hiding-ratio}
\end{equation}

Let $\mathrm{LO}$ denote the class of nonadaptive local measurements, $\mathrm{LOCC}$ the class of measurements implementable by local operations and classical communication, $\mathrm{LOCC}_\to$ its one-way subclass, and $\mathrm{SEP}$ the class of separable measurements. The inclusions $\mathrm{LO}\subseteq\mathrm{LOCC}_\to\subseteq\mathrm{LOCC}\subseteq\mathrm{SEP}$ give
\begin{equation}
  R_{\mathrm{SEP}}(n,m)
  \le R_{\mathrm{LOCC}}(n,m)
  \le R_{\mathrm{LOCC}_\to}(n,m)
  \le R_{\mathrm{LO}}(n,m).
  \label{eq:measurement-ratio-chain}
\end{equation}
The SWAP operator gives $R_{\mathrm{SEP}}(n,m)\ge d$, where $d=\min\{n,m\}$ \cite{LPW18}. On the other hand, Corr{\^e}a, Lami, and Palazuelos proved $R_{\mathrm{LO}}(n,m)\le2\sqrt2\,d$ \cite{CLP22}. Thus
\begin{equation}
  d
  \le R_{\mathrm{SEP}}(n,m)
  \le R_{\mathrm{LOCC}}(n,m)
  \le R_{\mathrm{LOCC}_\to}(n,m)
  \le R_{\mathrm{LO}}(n,m)
  \le 2\sqrt2\,d.
  \label{eq:previous-data-hiding-chain}
\end{equation}

As an application of our results, we improve the upper bound to $\sqrt{2}d$.

\begin{proposition}\label{prop:lo-dominates-injective}
For every $z\in M_n\otimes M_m$, we have\footnote{A similar bound for the self-adjoint injective norm was proved in \cite[Proposition~22]{LPW18}. Our bound controls the complex injective norm directly and thus avoids the extra factor $\sqrt2$ used in \cite{CLP22}, due to \cite{RV15}, to compare the complex and self-adjoint injective norms.
}:
\begin{equation}
  \norm{z}_\varepsilon\le\norm{z}_{\mathrm{LO}}.
  \label{eq:lo-dominates-injective}
\end{equation}
\end{proposition}
\begin{proof}
By the fact that contraction operators are linear combinations of unitaries, the supremum defining $\norm{z}_\varepsilon$ may be taken over unitaries.
For any unitaries $u\in U(n)$ and $v\in U(m)$, we take the spectral decompositions:
\begin{equation}
  u=\sum_\alpha\lambda_\alpha E_\alpha,
  \qquad
  v=\sum_\beta\mu_\beta F_\beta,
  \qquad
  \abs{\lambda_\alpha}=\abs{\mu_\beta}=1,
\end{equation}
where $(E_\alpha)$ and $(F_\beta)$ are rank-1 projectors. Then
\begin{equation}
    \abs{\tr\qty((u\otimes v)z)}
    =
    \abs{\sum_{\alpha,\beta}\lambda_\alpha\mu_\beta
    \tr\qty((E_\alpha\otimes F_\beta)z)} 
    \le
    \sum_{\alpha,\beta}
    \abs{\tr\qty((E_\alpha\otimes F_\beta)z)}
    \le\norm{z}_{\mathrm{LO}}.
\end{equation}
Taking the supremum over $u$ and $v$ proves the claim.
\end{proof}

Define
\begin{equation}
  R_\varepsilon(n,m)
  =
  \sup_{0\ne h=h^*}
  \frac{\norm{h}_1}{\norm{h}_\varepsilon}.
  \label{eq:injective-data-hiding-ratio}
\end{equation}
\Cref{prop:lo-dominates-injective,thm:main} give
\begin{equation}
  R_{\mathrm{LO}}(n,m)
  \le R_\varepsilon(n,m)
  \le \sqrt2\,\min\{n,m\}.
  \label{eq:lo-data-hiding-upper-bound}
\end{equation}
Together with \cref{eq:measurement-ratio-chain} and the SWAP lower bound, and suppressing the common argument $(n,m)$, this yields
\begin{equation}
  d
  \le R_{\mathrm{SEP}}
  \le R_{\mathrm{LOCC}}
  \le R_{\mathrm{LOCC}_\to}
  \le R_{\mathrm{LO}}
  \le R_\varepsilon
  \le \sqrt2\,d.
  \label{eq:data-hiding-chain}
\end{equation}

\section{Conclusion and outlook}\label{sec:outlook}

We proved that the optimal universal constant in the comparison between the trace norm and the injective tensor norm associated with trace norms is $\sqrt{2}$. The upper bound follows from an $L_1$ noncommutative Khintchine inequality whose random coefficients are the entries of a Haar unitary, while the lower bound is obtained from a construction based on the fermionic canonical anticommutation relations (CAR). The same asymptotic constant remains sharp for Hermitian tensors with vanishing partial traces. The bound also yields upper bounds for quantum data-hiding against local operations and other natural restricted classes of POVMs.

Several questions remain. 
One may replace the trace norm by other Schatten norms and ask for the corresponding sharp comparison with product observables. 
For multipartite systems, one may compare global observables with fully product observables and ask for the sharp dependence of the optimal bound on the number of subsystems and their local dimensions.

The dimensional requirements of the lower bound are less clear. Our construction uses a space of dimension $2^{n^2}$. It would also be interesting to determine whether the second local dimension can be reduced (for example, to a polynomial in $n$) while retaining the asymptotic constant $\sqrt{2}$.
An extremal version of this question is to determine the value of \(C_{n,n}\). More generally, if we impose the restriction $m/n\to\lambda$, how does the limiting constant depend on $\lambda$?  

\section*{Acknowledgements and AI use}
We thank Sergey Bravyi for encouraging discussions. ZL thanks Raz Firanko and Timothy Hsieh for a previous collaboration \cite{li2026unified} that sparked his interest in the question studied here. XS acknowledges financial support from NSERC Discovery Grant No. RGPIN-2024-04569 and No. DGECR-2024-004.

We used AI tools (ChatGPT and Claude) to draft the manuscript (later extensively revised by humans) and to develop the Lean formalization. 
We manually verified that the Lean statements match the corresponding statements in the paper.
These tools were also used for preliminary numerics (not reported here) to test constructions before proving the results. The main findings are human-generated.

\appendix
\section{Missing proofs for the lower bound}
\phantomsection\label{app:proof-square-zero}
\begin{proof}[Proof of \Cref{prop:square-zero}]

Let $C=(c_{ij})_{i,j=1}^n$ and $R=C^*C$, so that $R_{ab}=\sum_i c_{ia}^*c_{ib}$. 

First, we claim that
\begin{align}
  &\comm{R_{ab}}{c_{kl}}=-\delta_{al}c_{kb},\qquad   \acomm{c_{ij}}{c_{kl}}=0,  \label{eq:Cadmissible}\\
  &\comm{R_{ab}}{R_{cd}}=\delta_{bc}R_{ad}-\delta_{ad}R_{cb}.   \label{eq:R-commutator}
\end{align}
The first equation follows from the CAR and the identity $[XY,Z]=X\{Y,Z\}-\{X,Z\}Y$; the second is one of the CAR relations; the third follows from the first and the identity $[X,YZ]=[X,Y]Z+Y[X,Z]$.

Second, we claim that the properties \cref{eq:Cadmissible} are closed under certain multiplication.
More precisely, we claim that if a block matrix $X=(X_{ij})$ satisfies:
\begin{equation}
  \comm{R_{ab}}{X_{kl}}=-\delta_{al}X_{kb},
  \qquad
  \acomm{X_{ij}}{X_{kl}}=0,
  \label{eq:admissible-matrix}
\end{equation}
then the above equations also hold when $X$ is replaced by $X(R-\lambda I)$ for any $\lambda\in \mathbb C$.

For convenience, denote $Y=XR$. Using \cref{eq:admissible-matrix,eq:R-commutator}, we obtain
\begin{equation}
\begin{aligned}
  \comm{R_{ab}}{Y_{kl}}
  &=\sum_t\comm{R_{ab}}{X_{kt}}R_{tl}+\sum_tX_{kt}\comm{R_{ab}}{R_{tl}} \\
  &=-X_{kb}R_{al}+X_{kb}R_{al}-\delta_{al}\sum_tX_{kt}R_{tb} \\
  &=-\delta_{al}Y_{kb}.
\end{aligned}
\end{equation}
Thus $Y$, and hence $Y-\lambda X$, satisfies the covariance relation in \cref{eq:admissible-matrix}. 
For the anticommutation relation, we have
\begin{align}\label{eq:expand}
  \acomm{Y_{ij}-\lambda X_{ij}}{Y_{kl}-\lambda X_{kl}}
  =\acomm{Y_{ij}}{Y_{kl}}
  -\lambda\qty(\acomm{Y_{ij}}{X_{kl}}+\acomm{X_{ij}}{Y_{kl}})
  +\lambda^2\acomm{X_{ij}}{X_{kl}}.
\end{align}
The quadratic term vanishes due to \cref{eq:admissible-matrix}.
For the linear term, using \cref{eq:admissible-matrix} and the identity $\{AB,C\}=A[B,C]+\{A,C\}B$, we get:
\begin{equation}
  \acomm{Y_{ij}}{X_{kl}}=-X_{il}X_{kj},
  \qquad
  \acomm{X_{ij}}{Y_{kl}}=-X_{kj}X_{il}.
  \label{eq:mixed-anticommutators}
\end{equation}
Their sum vanishes because the entries of $X$ anticommute (\cref{eq:admissible-matrix}). 
For the first term in \cref{eq:expand}, a similar calculation gives
\begin{equation}
\begin{aligned}
  \acomm{Y_{ij}}{Y_{kl}}
  =&\sum_{a,b}X_{ia}X_{kb}\comm{R_{aj}}{R_{bl}}
  +\sum_{a,b}X_{ia}\comm{R_{aj}}{X_{kb}}R_{bl} 
  +\sum_{a,b}X_{kb}\comm{R_{bl}}{X_{ia}}R_{aj} \\
  =&-\sum_a\acomm{X_{il}}{X_{ka}}R_{aj}
  =0.
\end{aligned}
\end{equation}
Therefore, \cref{eq:expand} vanishes.
This proves the claim.

Now we use this closure property inductively.
If $q$ is nonzero ($q=0$ is trivial), factor it over $\mathbb C$ as $q(t)=\alpha\prod_{\nu=1}^{\deg q}(t-\lambda_\nu)$. 
Starting with $X_0=\alpha C$, define
\begin{equation}
  X_\nu=X_{\nu-1}(R-\lambda_\nu I),
  \qquad 1\leq\nu\leq \deg q.
\end{equation}
The closure property then shows that each $X_\nu$, in particular $z = X_{\deg q}=Cq(R)$, satisfies \cref{eq:admissible-matrix}. 

Writing $z$ as $(Q_{ij})$, we have $T_q(a)=\sum_{i,j}a_{ij}Q_{ij}$. Therefore, for every $a\in M_n$,
\begin{equation}
  T_q(a)^2
  =\frac{1}{2} \sum_{i,j,k,l} a_{ij}a_{kl}\acomm{Q_{ij}}{Q_{kl}}
  =0.
\end{equation}
This concludes the proof.
\end{proof}

\phantomsection\label{app:proof-quarter-circle}
\begin{proof}[Proof of \Cref{thm:quarter-circle}]

For convenience, we pass to the squared singular values. Denote $A=\frac{2}{n}C^*C$.
The eigenvalues of $A$ are $\sigma_k^2$, so we define their empirical distribution by
\begin{equation}
  \mu_n
  =
  \frac1{nD_n}
  \sum_{k=1}^{nD_n}\delta_{\sigma_k^2}.
\end{equation}
It is enough to prove that $\mu_n$ converges weakly to the pushforward of $\nu_{\mathrm{QC}}$ under $\sigma\mapsto\sigma^2$, which is a Marchenko--Pastur law $\mu_{\mathrm{MP}}$ \cite{MP67}, whose moments are the Catalan numbers:
\begin{equation}
  \int_0^4x^p\,d\mu_{\mathrm{MP}}(x)
  =
  \int_0^2\sigma^{2p}\,d\nu_{\mathrm{QC}}(\sigma)
  =
  C_p
  =
  \frac1{p+1}\binom{2p}{p}. \label{eq:4.19}
\end{equation}
Since $\mu_{\mathrm{MP}}$ has compact support, it is uniquely determined by its moments \cite{MS17}. 
It therefore suffices to prove that, for every $p\geq0$,
\begin{equation}
  \lim_{n\to\infty}\int x^p\,d\mu_n(x)=C_p.
  \label{eq:MP-moment-target}
\end{equation}
The case $p=0$ is immediate.
In the following, we assume $p\geq1$ and prove it by explicit calculation.

The block entries of $A$ are
\begin{equation}
  A_{ab} = \frac2n \sum_i c_{ia}^*c_{ib}. \label{eq:4.20}
\end{equation}
Let $\tau=\frac1{D_n}\tr$ be the normalized trace on the Fock space $\mathcal F_n$. The CAR implies that
\begin{equation}
  \tau(c_\alpha^*c_\beta) = \tau(c_\beta c_\alpha^*) = \frac12\delta_{\alpha\beta}. \label{eq:4.2}
\end{equation}
Expanding the trace gives
\begin{equation}
  \int x^p\,d\mu_n(x)
  =
  \frac1{nD_n}\tr(A^p)
  =
  \frac{2^p}{n^{p+1}}
  \sum_{\substack{a_1,\ldots,a_p\\ i_1,\ldots,i_p}}
  \tau\qty(\prod_{r=1}^p c_{i_ra_r}^*c_{i_ra_{r+1}}).
  \label{eq:4.21}
\end{equation}
Here every $a_r$ and $i_r$ ranges from $1$ to $n$, and we use the cyclic convention $a_{p+1}=a_1$.

We now evaluate each summand in \cref{eq:4.21} using the fermionic Wick formula. 
In our case, we need to pair creation operators with annihilation operators.
Each pairing corresponds to a permutation $\pi\in S_p$, which pairs $c_{i_ra_r}^*$ with $c_{i_{\pi(r)}a_{\pi(r)+1}}$ ($1\leq r\leq p$), contributing
\begin{equation}
  \frac12\delta_{i_r,i_{\pi(r)}}\delta_{a_r,a_{\pi(r)+1}}.
\end{equation}
Let $\gamma=(1\,2\,\cdots\,p)$. With cyclic indexing, $a_{\pi(r)+1}=a_{\gamma\pi(r)}$, and hence
\begin{equation}
  \tau\qty(\prod_{r=1}^p c_{i_ra_r}^*c_{i_ra_{r+1}})
  =\frac1{2^p}\sum_{\pi\in S_p}\varepsilon(\pi)
  \prod_{r=1}^p\delta_{i_r,i_{\pi(r)}}\delta_{a_r,a_{\gamma\pi(r)}}.
  \label{eq:4.22}
\end{equation}
Here the Wick sign $\varepsilon(\pi)$ is determined by the parity of the crossing number of the associated pairing\footnote{The number of crossings in the diagram ($2p$ points, $p$ edges) obtained by joining each creation--annihilation pair.}.

Substituting \cref{eq:4.22} into \cref{eq:4.21} and exchanging the order of summation, the Kronecker constraints leave $\#\pi$ free $i$-indices and $\#(\gamma\pi)$ free $a$-indices, where $\#\rho$ denotes the number of cycles of a permutation $\rho$, giving:
\begin{equation}
  \int x^p\,d\mu_n(x) = \sum_{\pi\in S_p} \varepsilon(\pi) n^{-p-1+\#\pi+\#(\gamma\pi)}, \label{eq:4.23}
\end{equation}

The exponent in \cref{eq:4.23} is standard in random matrix theory. It is nonpositive and vanishes if and only if $\pi^{-1}$ is noncrossing as a permutation\footnote{Namely, every cycle is clockwise, and there are no $a<b<c<d$ such that $a,c$ belong to one cycle and $b,d$ belong to a different cycle.}\cite{Bia97, MS17}.
There are $C_p$ such permutations.
Moreover, $\pi^{-1}$ being a noncrossing permutation implies the pairing associated with $\pi$ is noncrossing, hence the Wick sign $\varepsilon(\pi)=1$. 
All remaining permutations are of lower order. Therefore,
\begin{equation}
  \int x^p\,d\mu_n(x) = C_p+O_p(n^{-1}). \label{eq:4.24}
\end{equation}
This proves \cref{eq:MP-moment-target}, and the method of moments gives $\mu_n\Rightarrow\mu_{\mathrm{MP}}$, and hence $\nu_n\Rightarrow\nu_{\mathrm{QC}}$.
\end{proof}

\bibliographystyle{alpha} \bibliography{global_vs_product}
\end{document}